\documentclass{article}

\usepackage[dvips]{graphicx}
\usepackage[latin1]{inputenc}
\usepackage{amssymb,amsmath,amsfonts}
\usepackage{color}
\usepackage{hyperref,breakurl}
\usepackage{floatflt,float}
\usepackage{paralist}

\usepackage{amsthm}

\usepackage[ruled]{algorithm2e}

\newcommand{\N}{\mathbb{N}}

\newcommand{\R}{\mathbb{R}}

\newcommand{\Order}{\mathcal{O}}
\newcommand{\Normal}{\mathcal{N}}

\newcommand{\id}{\mathbb{I}}
\newcommand{\ones}{\mathbf{1}}

\newcommand{\bmat}{\begin{pmatrix}}
\newcommand{\emat}{\end{pmatrix}}

\newtheorem{theorem}{Theorem}

\begin{document}

\begin{center}
\textbf{\LARGE A Fast Incremental BSP Tree Archive\\[0.5em]for Non-dominated Points%
}\\[2em]
Tobias Glasmachers\\
Institute for Neural Computation, Ruhr-University Bochum, Germany\\
\texttt{tobias.glasmachers@ini.rub.de}\\[2em]
Note: A paper with identical content was submitted to EMO'2017.\\[1em]
\hrulefill\\[2em]
\end{center}

\begin{abstract}
Maintaining an archive of all non-dominated points is a standard task in
multi-objective optimization. Sometimes it is sufficient to store all
evaluated points and to obtain the non-dominated subset in a
post-processing step. Alternatively the non-dominated set can be updated
on the fly. While keeping track of many non-dominated points efficiently
is easy for two objectives, we propose an efficient algorithm based on a
binary space partitioning (BSP) tree for the general case of three or
more objectives. Our analysis and our empirical results demonstrate the
superiority of the method over the brute-force baseline method, as well
as graceful scaling to large numbers of objectives.
\end{abstract}

\section{Introduction}
\label{section:introduction}

Given $m \geq 2$ objective functions $f_1, \dots, f_m : X \to \R$, a
central theme in multi-objective optimization is to find the Pareto set
of optimal compromises: the set of points $x \in X$ that cannot be
improved in any objective without getting worse in another one. The
cardinality of this set of often huge or even infinite. In a black-box
setting it is best approximated by the set of mutually non-dominated
query points.

An archive $A$ of non-dominated solutions is a data structure for
keeping track of the known non-dominated points
$\{x^{(1)}, \dots, x^{(n)}\} \subset X$
of a multi-objective optimization problem. It can result from a single
run of an optimization algorithm, or from many runs of potentially
different algorithm.

Depending on the application such an archive can serve different
purposes. It can act as a portfolio of solutions accessed by a decision
maker, possibly after going through further post-processing steps. It
can also act as input to various algorithms, e.g., selection operators
of evolutionary multi-objective optimization algorithms (MOEAs),
stopping criteria, and performance assessment and monitoring tools. All
of these algorithms may involve the computation of set quality
indicators such as dominated hypervolume, for which the computation of
the non-dominated subsets is a pre-processing step.

Some of the above applications require access to the non-dominated set
``anytime'', i.e., already during the optimization (online case, called
dynamic problem in \cite{schuetze:2003}), while others get along with
storing all solutions and extracting the non-dominated subset after the
optimization is finished (offline or static case). The storage of
millions of intermediate points does not pose a problem on today's
computers, and even full non-dominated sorting is feasible on huge
collections of objective vectors with suitable algorithms. Hence, we
consider the offline case a solved problem, at least for moderate values
of~$m$.

In this paper we are interested in archives with the following properties:
\begin{compactitem}
\item[$\bullet$]
	Online updates of the set of non-dominated points shall be efficient.
	In particular, it shall be feasible to process long sequences of
	candidate solutions one by one, i.e., as soon as they are proposed
	and evaluated by an iterative optimization algorithm.
\item[$\bullet$]
	The archive shall not contain dominated points, even if these points
	were non-dominated at an earlier stage. I.e., points shall be
	removed as soon as they are dominated by a newly inserted point.
\item[$\bullet$]
	The full set of non-dominated points shall be stored, not an
	approximation set of a-priori bounded cardinality.
\item[$\bullet$]
	Ideally, the algorithm should scale well not only to large archives,
	but also to a large number of objectives.
\end{compactitem}
These prerequisites imply the following processing steps. First it is
checked whether a candidate point $x$ is dominated by any point in the
archive or not. If $x$ is non-dominated then it is added to the archive.
In addition, any points in $A$ that happen to be dominated by $x$ are
removed. All of these steps should be as efficient as possible. At the
very least, they should be faster than the $\Order(n m)$ linear ``brute
force'' search through the archive.

In this paper we propose an algorithm achieving this goal. It is based
on a binary space partitioning tree (BSP tree). Its performance is
demonstrated empirically, for large archives (up to $n = 2^{18}$
non-dominated points) and large numbers of objectives (up to $m = 50$),
as well as analytically in the form of asymptotic (lower and upper)
runtime bounds. We provide C++ source code for the proposed archive.%
\footnote{\url{http://www.ini.rub.de/PEOPLE/glasmtbl/code/ParetoArchive/}}
It is based on the implementation from the (non-public) code base of the
black box optimization competition (BBComp)%
\footnote{\url{http://bbcomp.ini.rub.de}},
where it is applied for online monitoring of the optimization progress.

In the remainder of this paper we first define the problem and fix our
notation. After reviewing related work we present the proposed archiving
algorithm. We derive asymptotic lower and upper bounds on its runtime
and assess its practical performance empirically on a variety of tasks.

\section{Definitions and Notation}

\paragraph{Dominance Order}
The objectives are collected in the vector-valued objective function
$f : X \to \R^m$, $f(x) = \big( f_1(x), \dots, f_m(x) \big)$. For
objective vectors $y, y' \in \R^m$ we define the Pareto dominance
relation
\begin{align*}
	y &\preceq y' \quad \Leftrightarrow \quad y_k \leq y'_k \text{ for all } k \in \{1, \dots, m\} \enspace, \\
	y &\prec y' \quad \Leftrightarrow \quad y \preceq y' \text{ and } y \not= y' \enspace.
\end{align*}
This relation defines a partial order on $\R^m$, incomparable values
$y, y'$ fulfilling $y \not\preceq y'$ and $y' \not\preceq y$ remain.
The relation is pulled back to the search space~$X$ by the definition
$x \preceq x'$ iff $f(x) \preceq f(x')$.

\paragraph{Pareto front and Pareto set}
Let $Y = \{f(x) \,|\, x \in X\} = f(X) \subset \R^m$ denote the image of
the objective function (also called the attainable objective space).
The Pareto front is defined as the set of objective values that are
optimal w.r.t.\ Pareto dominance, i.e., the set of non-dominated
objective vectors
\begin{align*}
	Y^* = \Big\{ y \in Y \,\Big|\, \not\exists \, y' \in Y: y' \prec y \Big\}
	\enspace.
\end{align*}
The Pareto set is $X^* = f^{-1}(Y^*)$.

\paragraph{Non-dominated points}
In the sequel we will deal with objective vectors, not with actual
search points. In this paper we restrict ourselves to maintaining only a
single search point for each non-dominated objective vector, although it
is possible to observe any number of search points of equal quality.
In other words, we aim to approximate the Pareto front, represented by a
minimal complete Pareto set.

For our purposes, let $y^{(1)}, \dots, y^{(n)} \in \R^m$ denote the set
of all known non-dominated objective vectors stored in the archive at
some point, with $n$ denoting the current cardinality of the archive.
The objective vector of the new candidate $x$ is denoted by
$y = f(x) \in \R^m$.

\section{Related Work}
\label{section:related-work}

\paragraph{Fixed Memory Approximations.}
The growth of the set of non-dominated points is in general unbounded.
Still, many MOEAs use their (fixed size) population as a rough
approximation of the Pareto front. Even if an external archive is
employed, computer memory is finite, which sets limits to the approach
of archiving an unbounded number of non-dominated points. Therefore
limited memory archives (e.g., based on clustering) were subject of
intense study~\cite{fieldsend:2003,lopez:2011}. Their main disadvantage
is the limited precision of their Pareto front representation, which can
even lead to oscillatory behavior when used for
optimization~\cite{fieldsend:2003}.

\paragraph{Memory Consumption.}
Nowadays even commodity PCs are equipped with gigabyes of RAM,
enabling the storage of millions of objective vectors in memory.%
\footnote{This is usually sufficient for the needs of evolutionary
  optimization. In data base query problems larger sets must be
  processed. Hence in some applications memory consumption is still a
  concern.}
This makes it feasible for MOEAs to maintain all non-dominated points in
the population, as proposed by Krause et al.~\cite{krause:2016}.
Although the algorithm is limited to $m=2$ objectives, it demonstrates
successfully that even on standard hardware memory is no longer a
limiting factor for archiving of all known non-dominated points.

\paragraph{Offline Case.}
Efficient algorithms are known for the offline case of obtaining the
Pareto optimal subset from a set of points, i.e., at the end of an
optimization run. This problem was first addressed in \cite{kung:1975}.
Even full non-dominated sorting (delivering not only the Pareto optimal
subset, but a partitioning into disjoint fronts) can be achieved in only
$\Order(n \log(n)^{m-1})$ operations and $\Order(n m)$
memory~\cite{jensen:2003,fortin:2013}. However, computing the full
Pareto front from scratch after every update is wasteful, and hence
specialized updating algorithm are needed for the online case.

\paragraph{Skylines.}
A closely related but slightly different problem is found in data base
queries: so-called ``skyline queries'' ask for a skyline of records,
which is just a different terminology for the non-dominated (Pareto
optimal) subset with respect to a specified subset of fields. While the
offline case is the most important one, online algorithms with good
anytime performance are of relevance.
The requirements differ from the problem under study in a decisive point:
in the data base setting the full set of records is available from the
start. The process is limited only by computational and memory
constraints, but not by the sequential nature of proposing points in an
iterative optimizer. Efficient algorithms (e.g., based on nearest
neighbor queries) exist for the skyline
problem~\cite{kossmann:2002,papadias:2003}.

\paragraph{Search Trees.}
Tree data structures of various sorts are attractive for tackling our
problem. They offer a natural problem decomposition, i.e., adding a
single point usually affects only a small neighborhood of the current
front, possibly represented by a small sub-tree.
The dominance decision tree (DDT) data structure was proposed
specifically for the problem at hand \cite{schuetze:2003}. They work
well if the fraction of non-dominated points and hence the archive is
small~\cite{schuetze:2003}. In contrast, quad trees can reduce the
computational effort for large archives \cite{mostaghim:2002}. A clear
disadvantage of quad trees generalized to $m > 2$ objectives is that
partitioning a space cell results in $2^m$ sub-cells, which limits the
approach to small numbers $m$ of objectives.

\section{Algorithms}
\label{section:algorithms}

We first discuss the trivial baseline method and a highly efficient
alternative for the special case of $m=2$ objectives. Then we present
our core contribution, an efficient method for the general case of
$m \geq 3$ objectives.

\subsection{Baseline Method}

The naive ``brute force'' baseline method stores all non-dominated
points in a flat linear memory vector (dynamically extensible array) or
linked list. The insert operation loops over the archive. It compares
the candidate vector $y$ to each stored $y^{(k)}$ w.r.t.\ dominance.
This takes $\Order(m)$ operations per point in the archive.
If $y^{(k)} \preceq y$ then the point does not need to be archived and
the procedure stops. If $y \prec y^{(k)}$ then $y^{(k)}$ is removed.
In the list, removal takes $O(1)$ operations. In the vector we overwrite
the dominated point with the last point, which is then removed
($\Order(m)$ operations). Finally, $y$ is appended to the end of the
list or vector (amortized constant time for vectors with a doubling
strategy, but even a possible $\Order(n m)$ relocation of the memory
block does affect the analysis).
Hence, in the worst case the brute force method performs $\Theta(n m)$
operations. It is significantly faster (in expectation over a presumed
random order of the archive) only if $y$ is dominated by more than
$\Order(1)$ points from the archive. Advantages of this algorithm are
the trivial implementation and, in case of a linear memory array,
optimal use of the processor cache.
The method can be considered an online variant of Deb's fast
non-dominated sorting algorithm \cite{deb:2000}, however, restricted to
the first non-dominance rank.

\subsection{Special Case of Two Objectives}

For the bi-objective case the Pareto front is well known to obey a
special structure that can be exploited for our purpose: we keep the
archive sorted w.r.t.\ the first objective $f_1$ in ascending order,
which automatically keeps it sorted in descending order w.r.t.\ to the
second objective $f_2$. Given $y$ we search for the indices
\begin{align*}
	\ell = \arg \max \Big\{j \,\Big|\, y^{(j)}_1 \leq y_1 \Big\}
	\qquad \text{and} \qquad
	r = \arg \min \Big\{j \,\Big|\, y^{(j)}_1 \geq y_1 \Big\}
\end{align*}
of $y$'s potential ``left'' and ``right'' neighbors on the front.
If the $\arg\min$ is undefined because the set is empty then $y$ is
non-dominated and we set $r = 0$ for further processing.
If $y$ is weakly dominated by any archived point then it is also weakly
dominated by $y^{(\ell)}$, which is the case exactly if
$y^{(\ell)}_2 \leq y_2$. Hence once $\ell$ is found the check is fast.
If $y$ happens to dominate any archived points, then these are of the
form $\{y^{(r)}, y^{(r+1)}, \dots, y^{(r+s)}\}$. Hence, given $r$ it is
easy to identify the dominated points, which are then removed from the
archive.

Self-balancing trees such as AVL-trees and red-black-trees are suitable
data structures for performing these operators quickly. The search for
$\ell$ and $r$ and the insert operation require $\Order(\log(n))$
operations each, and so does the removal of a point. In an amortized
analysis there can be at most one removal per insert operation, hence
the overall complexity remains as low as $\Order(\log(n))$, which is far
better than the $\Order(n)$ brute-force method (note that here $m$ does
not really enter the complexity since it is fixed to the constant $m=2$).
With this archive, the cumulated effort of $n$ iterative updates equals
(up to a constant) the full non-dominated sort post-processing step by
sweeping objectives~\cite{fortin:2013}.

\subsection{General Case of $m \geq 3$ Objectives}

For more than two objectives non-dominated sets obey far less structure
than in the bi-objective case. Hence it is unclear which speed-up is
achievable. However, it should be possible to surpass the linear
complexity of the brute-force baseline. In the following we present an
algorithm that achieves this goal.

We propose to store the archived non-dominated points in a binary
space partitioning (BSP) tree. Among the different variants a simple
k-d tree \cite{de-berg:2000} seem to be best suited. This choice was
briefly discussed and quickly dismissed in \cite{fieldsend:2003} and
\cite{schuetze:2003}. Nonetheless we propose an archiving algorithm
based on a k-d tree, for the following reason. We expect most new
non-dominated points to dominate only a local patch of the current front.
Therefore it should be possible to limit dominance comparisons to few
archived points close to the candidate point. Space partitioning is a
natural approach to exploiting this locality.

Let $R$ denote the root node. Each interior node of the tree is a data
structure keeping track of its left and right child nodes $\ell$ and $r$,
an objective index $j \in \{1, \dots, m\}$, and a threshold
$\theta \in \R$. For a node $N$ we refer to these data fields with the
dot-notation, i.e., $N.r$ refers to the right child node of~$N$. With
$p(N)$ we refer to the parent node, i.e., $p(N.\ell) = N$ and
$p(N.r) = N$, while $p(R)$ is undefined. We write $p^2(N) = p(p(N))$,
and $p^k(N)$ for the $k$-th ancestor of~$N$.

The objective space is partitioned as follows. Let $\Delta(N)$ denote
the subspace represented by the node $N$. We have $\Delta(R) = \R^m$.
We recursively define
$\Delta(N.\ell) = \{y \in \Delta(N) \,|\, y_{N.j} < N.\theta\}$ and
$\Delta(N.r) = \{y \in \Delta(N) \,|\, y_{N.j} \geq N.\theta\}$.

Each leaf node holds a set $P$ of objective vectors limited in
cardinality by a predefined bucket size~$B \in \N$.
In addition, each node (interior or leaf) keeps track of the number $n$
of objective vectors in the subspace it represents.

Newly generated objective vectors $y$ are added one by one to the
archive with the \texttt{Process} algorithm laid out in
algorithm~\ref{algo:Process}.

\begin{algorithm}
	\caption{Process($y$)}
	\label{algo:Process}
	$s \leftarrow$ CheckDominance($R$, $y$, $0$, $\emptyset$) \\
	\If {$s \geq 0$}
	{
		$N \leftarrow R$ \\
		\While {$N$ is interior node}
		{
			$N.n \leftarrow N.n + 1$ \\
			\textbf{if} {$y_{N.j} < N.\theta$} \textbf{then} $N \leftarrow N.\ell$ \textbf{else} $N \leftarrow N.r$ \\
		}
		$P \leftarrow N.P \cup \{y\}$ \\
		\textbf{if} {$|P| \leq B$} \textbf{then} $N.P \leftarrow P$; $N.n \leftarrow |P|$ \\
		\Else
		{
			make $N$ an interior node and create new leaf nodes as children \\
			select $N.j$ and $N.\theta$ (see text) \\
			$N.\ell.P \leftarrow \{p \in P \,|\, p_{N.j} < N.\theta\}$; $N.\ell.n \leftarrow |N.\ell.P|$ \\
			$N.r.P \leftarrow \{p \in P \,|\, p_{N.j} \geq N.\theta\}$; $N.r.n \leftarrow |N.r.P|$ \\
		}
	}
\end{algorithm}

\begin{algorithm}
	\caption{CheckDominance($N$, $y$, $B$, $W$)}
	\label{algo:CheckDominance}
	$k \leftarrow 0$ \\
	\If {$N$ is leaf node}
	{
		\ForEach {$p \in N.P$}
		{
			\textbf{if} {$p \preceq y$} \textbf{then} \Return {$-1$} \\
			\If {$y \prec p$}
			{
				$N.P \leftarrow N.P \setminus \{p\}$ \\
				$N.n \leftarrow N.n - 1$ \\
				$k \leftarrow k + 1$ \\
			}
		}
	}
	\If {$N$ is interior node}
	{
		\textbf{if} {$y_{N.j} < N.\theta$} \textbf{then} $B' \leftarrow B \cup \{N.j\}$; $W' \leftarrow W$ \\
		\textbf{else} $B' \leftarrow B$; $W' \leftarrow W \cup \{N.j\}$ \\
		\textbf{if} {$|W'| = m$} \textbf{then} \Return {$-1$} \\
		\textbf{else~if} {$|B'| = m$} \textbf{then} $k \leftarrow k + N.r.n$; $N.r.n \leftarrow 0$ \\
		\Else
		{
			\If {$W' = \emptyset \lor B = \emptyset$}
			{
				$s \leftarrow$ CheckDominance($N.\ell$, $y$, $B$, $W'$) \\
				\textbf{if} {$s < 0$} \textbf{then} \Return {$-1$} \\
				$k \leftarrow k + s$ \\
			}
			\If {$B' = \emptyset \lor W = \emptyset$}
			{
				$s \leftarrow$ CheckDominance($N.r$, $y$, $B'$, $W$) \\
				\textbf{if} {$s < 0$} \textbf{then} \Return {$-1$} \\
				$k \leftarrow k + s$ \\
			}
		}
		$N.n \leftarrow N.n - k$ \\
		\textbf{if} {$N.\ell.n > 0$ and $N.r.n = 0$} \textbf{then} overwrite $N$ with $N.\ell$ \\
		\textbf{if} {$N.\ell.n = 0$ and $N.r.n > 0$} \textbf{then} overwrite $N$ with $N.r$ \\
	}
	\Return {$k$} \\
\end{algorithm}

\subsubsection{Processing a Candidate Point.}

\texttt{Process} calls the recursive \texttt{CheckDominance} algorithm
(algorithm~\ref{algo:CheckDominance}), which returns the number of
points in the archive dominated by $y$, or $-1$ if $y$ is dominated by
at least one point in the archive. The procedure also removes all points
dominated by $y$. If $y$ is non-dominated (i.e., the return value is
non-negative), then $y$ is inserted into the tree by descending into the
leaf node $N$ fulfilling $y \in \Delta(N)$. If the insertion would
exceed the bucket size ($|P| > B$, with $P = N.P \cup \{y\}$), then the
node is split. To this end we need to select an objective index $j$ and
a threshold $\theta$. The only hard constraint on the choice of $j$ is
that $V_j = \{y_j \,|\, y \in P\}$ must contain at least two values, so
that splitting the space in between yields two leaf nodes holding at
most $B$ objective vectors. However, for reasons that will become clear
in the next section we prefer to select different objectives if possible
when descending the tree. For $j \in \{1, \dots, m\}$ we define the
distance $d_j = \min\{p^k(N).j = j \,|\, k \in \N\}$ of the node $N$ in
its chain of ancestors from the next node splitting at objective~$j$. We
set $d_j = \infty$ if the root node is reached before observing
objective~$j$.
We chose $j \in \arg\max_j \big\{ d_j \,\big|\, |V_j| > 1 \big\}$.
For the selection of $\theta$ we have to take into account that multiple
objective vectors may agree in their $j$-th component. We select
$\theta$ as the midpoint of two values from $V_j$ so that the split
balances the leaves as well as possible. For even $B$ (and hence an
uneven number of objective vectors) we prefer more points in the left
leaf.

\subsubsection{Recursive Check of the Dominance Relation.}

The \texttt{CheckDominance} algorithm is the core of our method. It
takes a node $N$, the candidate point $y$, and two index sets $B$ and
$W$ as input. It returns the number of points in the subtree strictly
dominated by $y$, and $-1$ if $y$ is weakly dominated by any point in
the space cell represented by the subtree. The algorithm furthermore
removes all dominated points from the subtree.

When faced with a leaf node it operates similar to the brute force
algorithm. However, for interior nodes it can do better. To this end,
note that the set $\Delta(N) \subset \R^m$ can be written in the form
\begin{align*}
	\Delta(N) = \Delta_1(N) \times \dots \times \Delta_m(N)
\end{align*}
where $\Delta_j(N) \subset \R$ is the projection of $\Delta(N)$ to the
$j$-th objective. Since the reals are totally ordered, we distinguish
the cases $y_j < \Delta_j(N)$, $y_j \in \Delta_j(N)$, and
$y_j > \Delta_j(N)$. Note that $y_j < \Delta_j(N)$ for all $j$ implies
that $y$ dominates the whole space cell $\Delta(N)$, similarly
$y_j > \Delta_j(N)$ implies that $y$ is dominated by any point in
$\Delta(N)$. If there exist $i, j$ so that $y_i < \Delta_i(N)$ and
$y_j > \Delta_j(N)$ then $y$ is incomparable to all points in
$\Delta(N)$.

The algorithm descends the tree by recursively invoking itself on the
left and right child nodes, but only if necessary. The recursion is
necessary only if the comparisons represented by the recursive calls up
the call stack do not determine the dominance relation between $y$ and
$\Delta(N)$ yet. At the root node we know that it holds
$y_j \in \Delta_j(R)$ for all $j \in \{1, \dots, m\}$. Hence we have
either $y_{R.j} \in \Delta_{R.\ell}$ or $y_{R.j} \in \Delta{R.r}$, and
hence either $y_{R.j} > \Delta_{R.\ell}$ or $y_{R.j} < \Delta{R.r}$.
The sets $B$ and $W$ keep track of the objectives in which $y$ is
better or worse than $\Delta(N)$. The two sets are apparently disjoint.
If they are non-empty at the same time then the candidate point and the
space cell are incomparable, hence the recursion can be stopped. If
$W$ equals the full set $\{1, \dots, m\}$ then $y$ is dominated by the
space cell $N$. The mere existence of the node guarantees that it
contains at least one point, so we can conclude that $y$ is dominated.
If on the other hand $B = \{1, \dots, m\}$ then all points in $N$ are
dominated are hence removed. If the algorithm finds one of its child
nodes empty after the recursion then it recovers a binary tree by
replacing the current node with the remaining child. No action is
required if both child nodes are empty since this implies $N.n = 0$,
and the node will be removed further up in the tree.

\subsubsection{Balancing the Tree.}

In contrast to the bi-objective case it is unclear how to balance the
tree at low computational cost.
This is not a severe problem for objective vectors drawn i.i.d.\ from
a fixed distribution. This situation is fulfilled in good enough
approximation when performing many short optimization runs. However, for
a single (potentially long) run of an optimizer we can expect a
systematic shift from low-quality early objective vectors towards better
and better solutions over time. Hence most points proposed late during
the run will tend to end up in the left child of a node, the split point
of which was determined early on. We counter this effect by introducing
a balancing mechanism as follows. If the quotient
$\frac{N.\ell.n}{N.r.n}$ raises above a threshold $z$ or falls below
$\frac{1}{z}$ then the smaller child node is removed, the larger one
replaces its parent, and the points represented by the smaller node are
inserted. Although this process is computationally costly, it can pay
off in the long run in case of highly unbalanced trees.

\section{Analysis}
\label{section:analysis}

In this section we analyze on the complexity of the BSP tree based
archive algorithm. We start with the storage requirements. When storing
$n$ non-dominated points, in the worst case there are $n$ distinct leaf
nodes, and hence $n-1$ interior nodes in the tree, requiring $\Order(n)$
memory in addition to the unavoidable requirement of $\Order(nm)$ for
storing the non-dominated objective vectors. Hence the added memory
footprint due to the BSP tree is unproblematic. In our implementation
the overhead of a tree node is 56 bytes on a 64bit system, which makes
it feasible to store millions of points in RAM.

The analysis of the runtime complexity is more involved. Since archiving
small numbers of points is uncritical, we focus on the case of large~$n$,
which is well described by an average case amortized analysis. For the
analysis we drop the rather heuristic balancing mechanism. For
simplicity we set the bucket size to $B = 1$ and assume a perfectly
balanced tree of depth $\log_2(n)$. Although optimistic, this assumption
is not too unrealistic: note that the depth of a random tree is
typically of order $2 \log_2(n)$.

The strongest technical assumption we make for the analysis is that
objective vectors are sampled i.i.d.\ from a static distribution. Let
$P$ denote a probability distribution on $\R^m$ so that for two random
objective vectors $a, b \sim P$ the events $a \preceq b$ and
$\exists j \in \{1, \dots, m\} : a_j = b_j$ have probability zero.
We consider a BSP archive constructed by inserting $n$ points sampled
i.i.d.\ from $P$. Then we are interested in bounding the expected
runtime $T$ required for processing a candidate point $y \sim P$.

For a node $N$ representing the space cell $\Delta(N)$ we define its
order w.r.t.\ the candidate point $y$ as
$k = |\{j \,|\, y_j \not\in \Delta_j(N) \}|$. We call a node comparable
to $y$ if its space cell contains at least one comparable point (w.r.t.\
the Pareto dominance relation). All incomparable cells are skipped by
the algorithm, hence the runtime is proportional to the number of
comparable nodes. A node is incomparable to $y$ if there exist $j_1$ and
$j_2$ such that $y_{j_1} < \Delta_{j_1}(N)$ and $y_{j_2} > \Delta_{j_2}(N)$.
Furthermore, cells dominated by $y$ ($y_j < \Delta_j(N)$ for all $j$)
don't need to be visited, actually, encountering such a cell stops the
algorithm immediately. Similarly, nodes dominated by $y$ can be ignored
in an amortized analysis since on average only one point can
be removed per insertion, and the cost of a removal is as low as
$\Order(\log(n))$. Hence all nodes of order $m$ can be ignored for the
analysis.

In the following we denote the probability for a random node at depth
$d$ below the root node (the root has depth $0$) to have order
$k \in \{0, \dots, m\}$ with $Q_d(k)$. We have $Q_0(0) = 1$ and
$Q_d(k) = 0$ for $k > d$.

The following theorem provides a lower bound on the runtime in the best
case, namely when each split of the BSP tree induces the same chance to
yield an incomparable child node.
\begin{theorem}
\label{theorem:lower}
Assume that when traversing from root to leaf no two space splits are
along the same objective.
Then we have $T \in \Omega(n^{\log_2(3/2)})$.
\end{theorem}
\begin{proof}
The prerequisite implies $m \geq d$ and hence $m \geq \log_2(n)$.
Since objective vectors are sampled i.i.d., the probability of the
candidate point to be covered by the left or right sub-tree is $50\%$
at each node. This corresponds to a $50\%$ chance to increment the order
$k$ when descending an edge of the tree, hence $k$ follows a binomial
distribution. We obtain $Q_d(k) = 2^{-d} \cdot \binom{d}{k}$ for $k \leq d$.
For given $k$, the chance of a node to be comparable to the candidate
point is $\min\{1, 2^{1-k}\}$, since for this to happen all $k$
decisions (descending left or right in the tree) must coincide. Hence
among the $2^d$ nodes at depth $d$ an expected number of
\begin{align*}
	1 + \sum_{k=1}^d 2^{1-k} \cdot \binom{d}{k} = 2^{d \cdot \log_2(3/2) + 1} - 1
\end{align*}
nodes is comparable to the candidate point. The statement follows by
summing over all depths $d \leq \log_2(n)$.
\hfill$\blacksquare$
\end{proof}
Under the milder (and actually pessimistic) assumption of random split
objectives we obtain a sub-linear upper bound.
\begin{theorem}
\label{theorem:upper}
Assume that each node splits the space along an objective $j$ drawn
uniformly at random from $\{1, \dots, m\}$. Then we have
$T \in o(nm)$.
\end{theorem}
\begin{proof}
In this case $Q$ is described by the following recursive formulas for $d > 0$:
\begin{align*}
	Q_d(0) & = \frac{1}{2} Q_{d-1}(0) \notag \\
	Q_d(k) & = \frac{1}{2} \left[ Q_{d-1}(k-1) \cdot \frac{m-k+1}{m} + Q_{d-1}(k) \cdot \frac{m+k}{m} \right]
\end{align*}
We obtain $\lim_{d \to \infty} Q_d(m) = 1$, hence only $o(n)$ nodes
are of order at most $m-1$, and only a subset of these must be visited.
Since processing a leaf node requires $\Order(m)$ operations in general,
we arrive at $T \in o(nm)$.
\hfill$\blacksquare$
\end{proof}

\section{Experimental Evaluation}
\label{section:experiments}

All three archives were implemented efficiently in C++.
Here we investigate how fast the archives operate in practice, how their
runtimes scale to large $n$ and $m$, and how their practical performance
relates to our analysis.

\subsection{Analytic Problems}

A first series of tests was performed with sequences of objective
vectors following controlled, analytic distributions. These tests allow
us to clearly disentangle effects caused by different fractions of
non-dominated points and non-stationarity due to improvement of
solutions over time.

\begin{figure}
\begin{center}
{
	\includegraphics[width=0.5\textwidth]{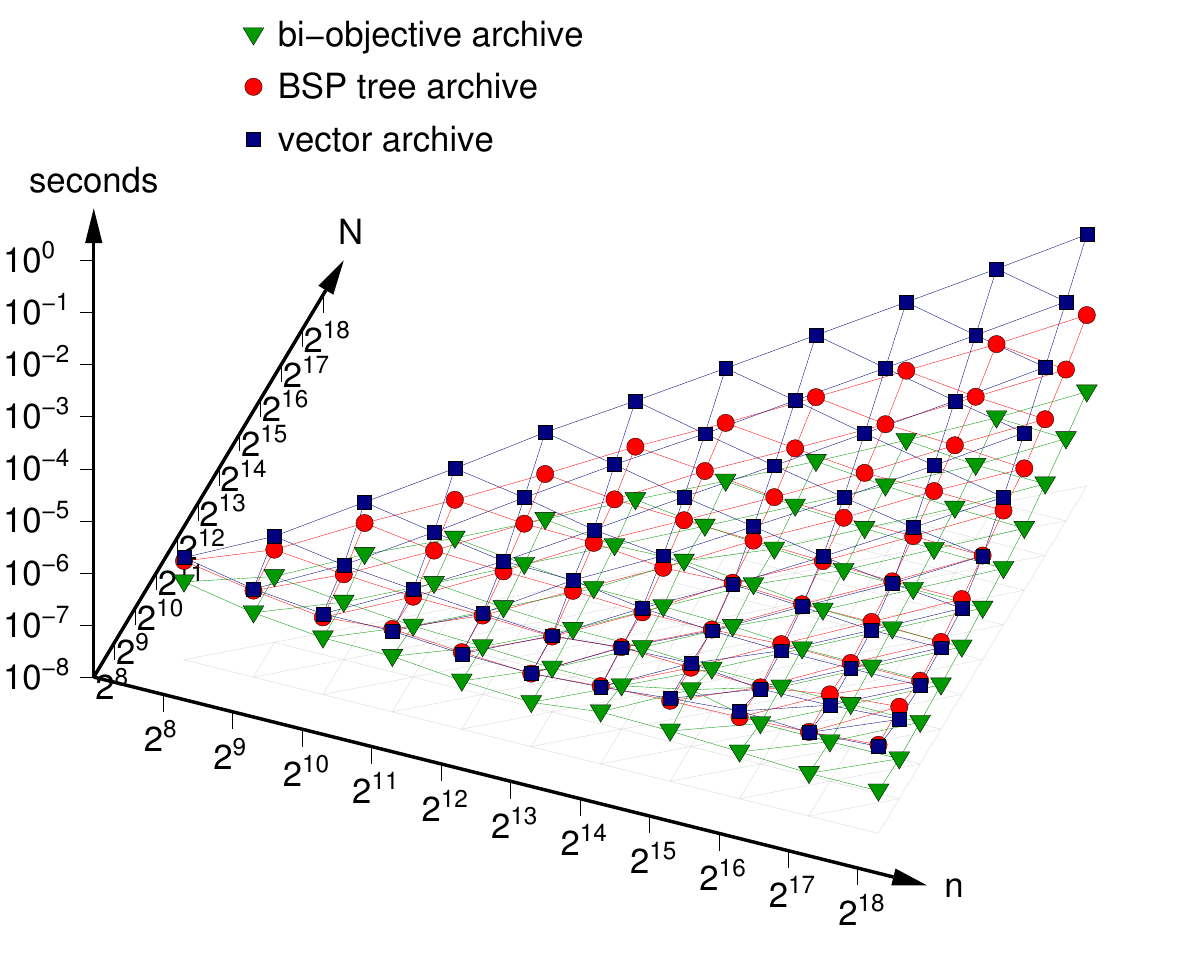}\includegraphics[width=0.5\textwidth]{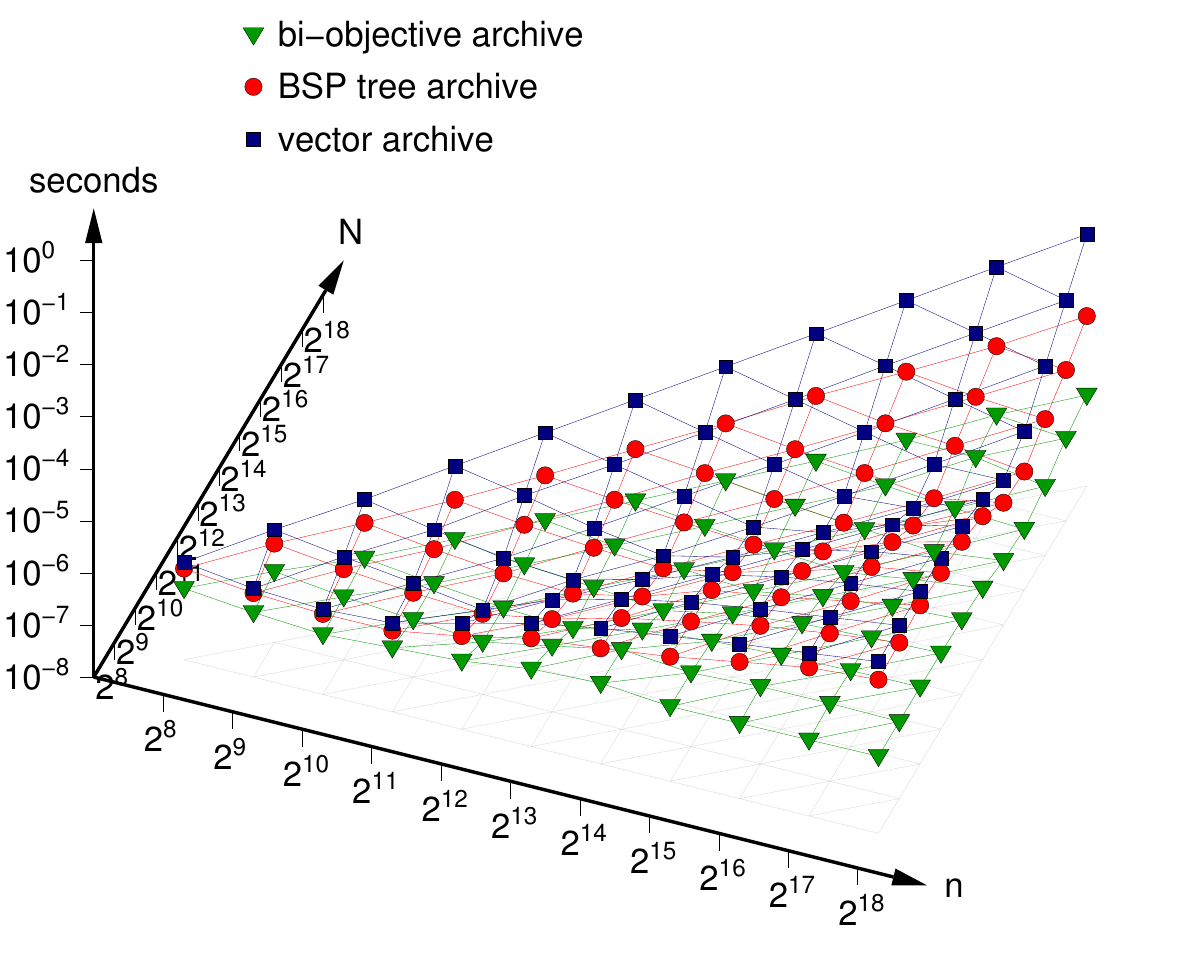}\\
	\includegraphics[width=0.5\textwidth]{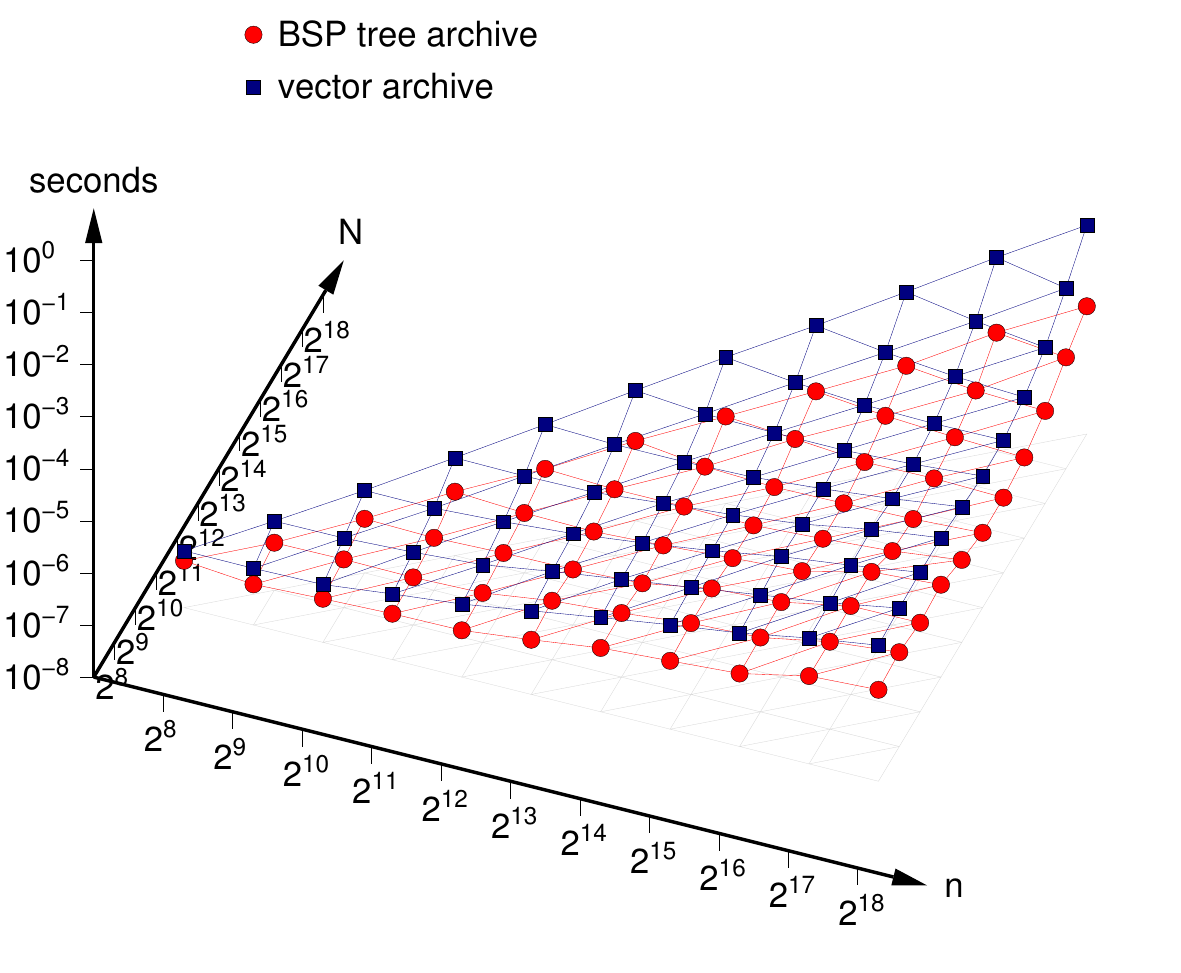}\includegraphics[width=0.5\textwidth]{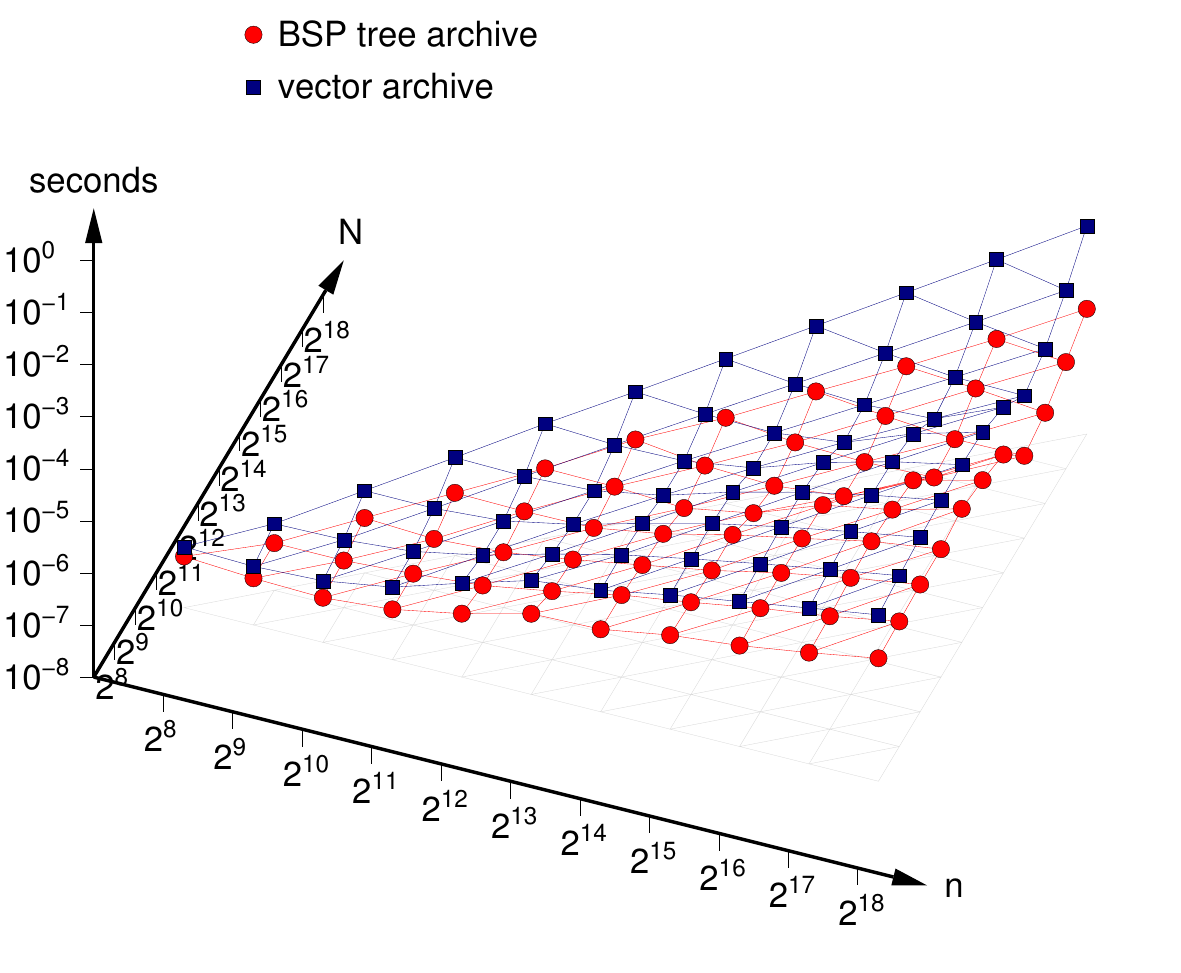}\\
}
\caption{ \label{figure:timings-1}
Processing time per objective vector for $m=2$ objectives (top) and
$m=3$ objectives (bottom), for a static distribution ($c=1$, left)
and solutions improving over time ($c=1.1$, right), for systematically
varied numbers of overall points ($n=N+D$) and non-dominated points ($N$)
in the range $2^8$ to $2^{18}$.
}
\end{center}
\end{figure}

\begin{figure}
\begin{center}
{
	\includegraphics[width=0.5\textwidth]{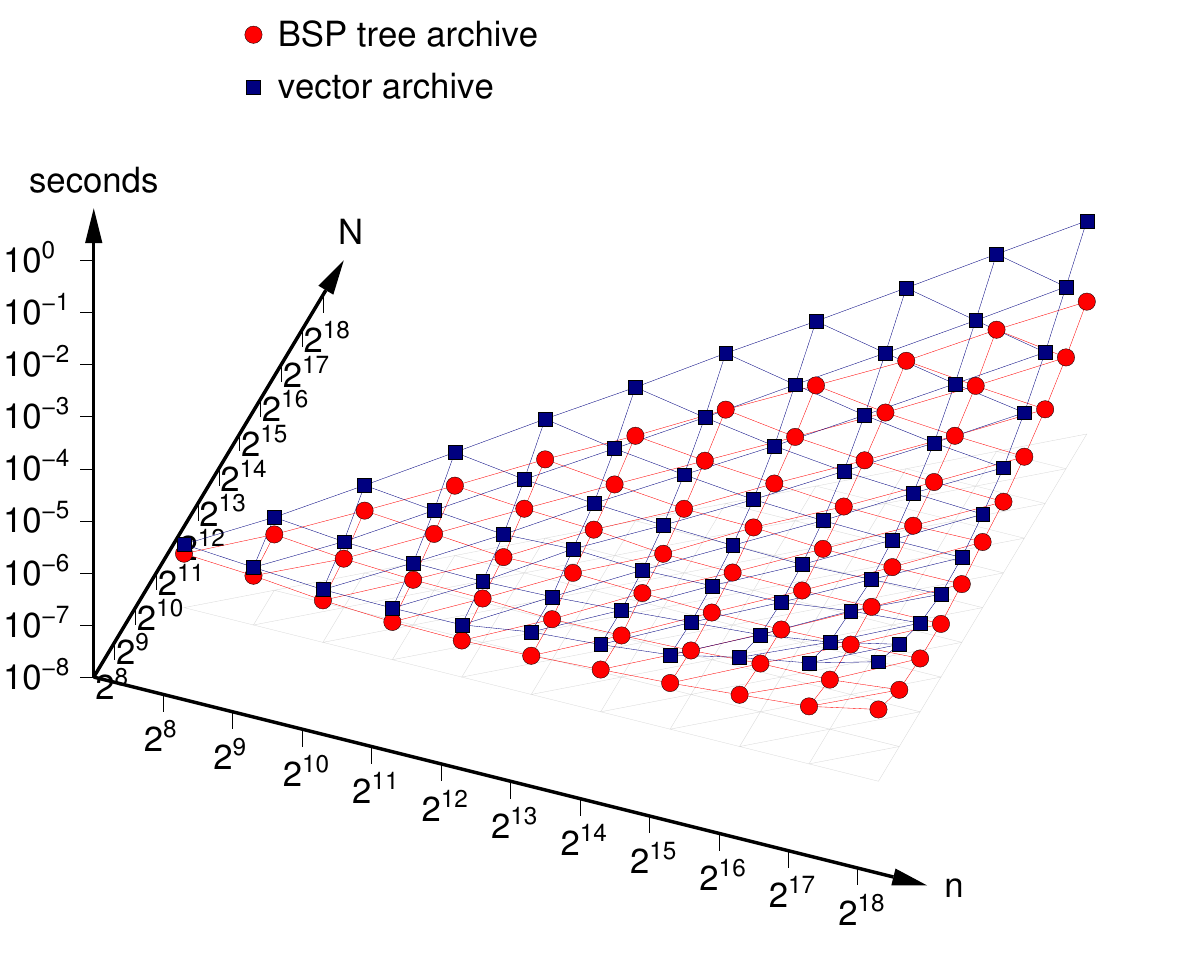}\includegraphics[width=0.5\textwidth]{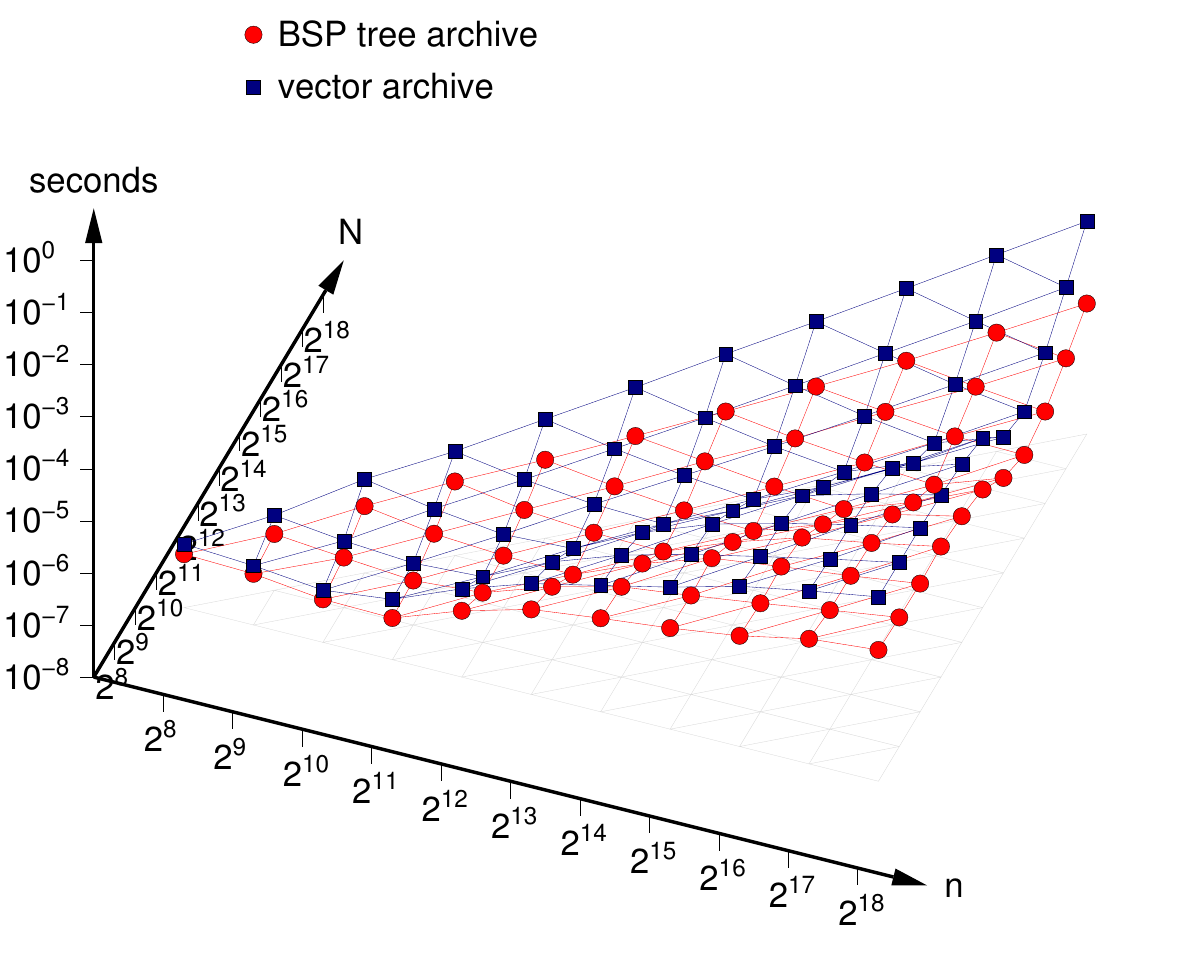}\\
	\includegraphics[width=0.5\textwidth]{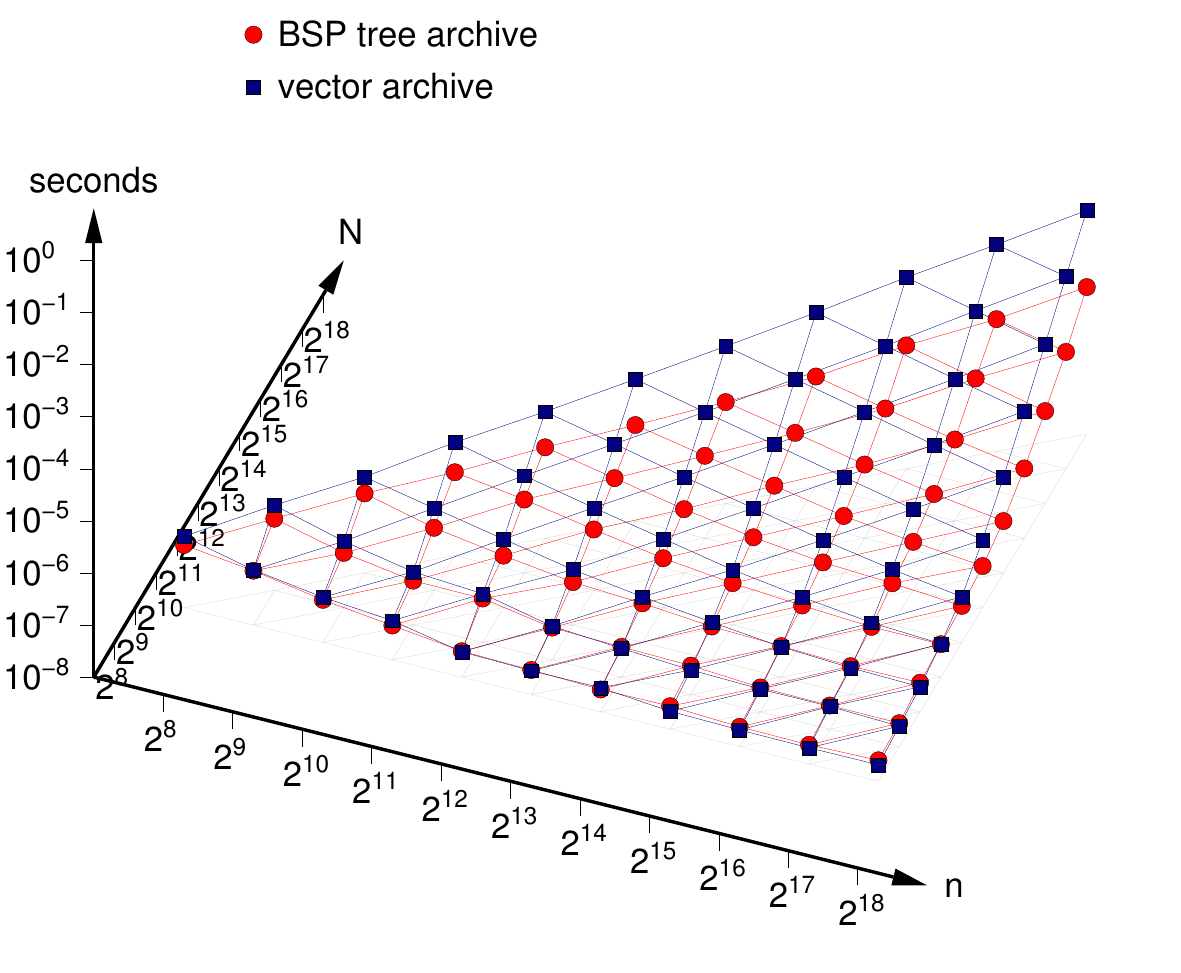}\includegraphics[width=0.5\textwidth]{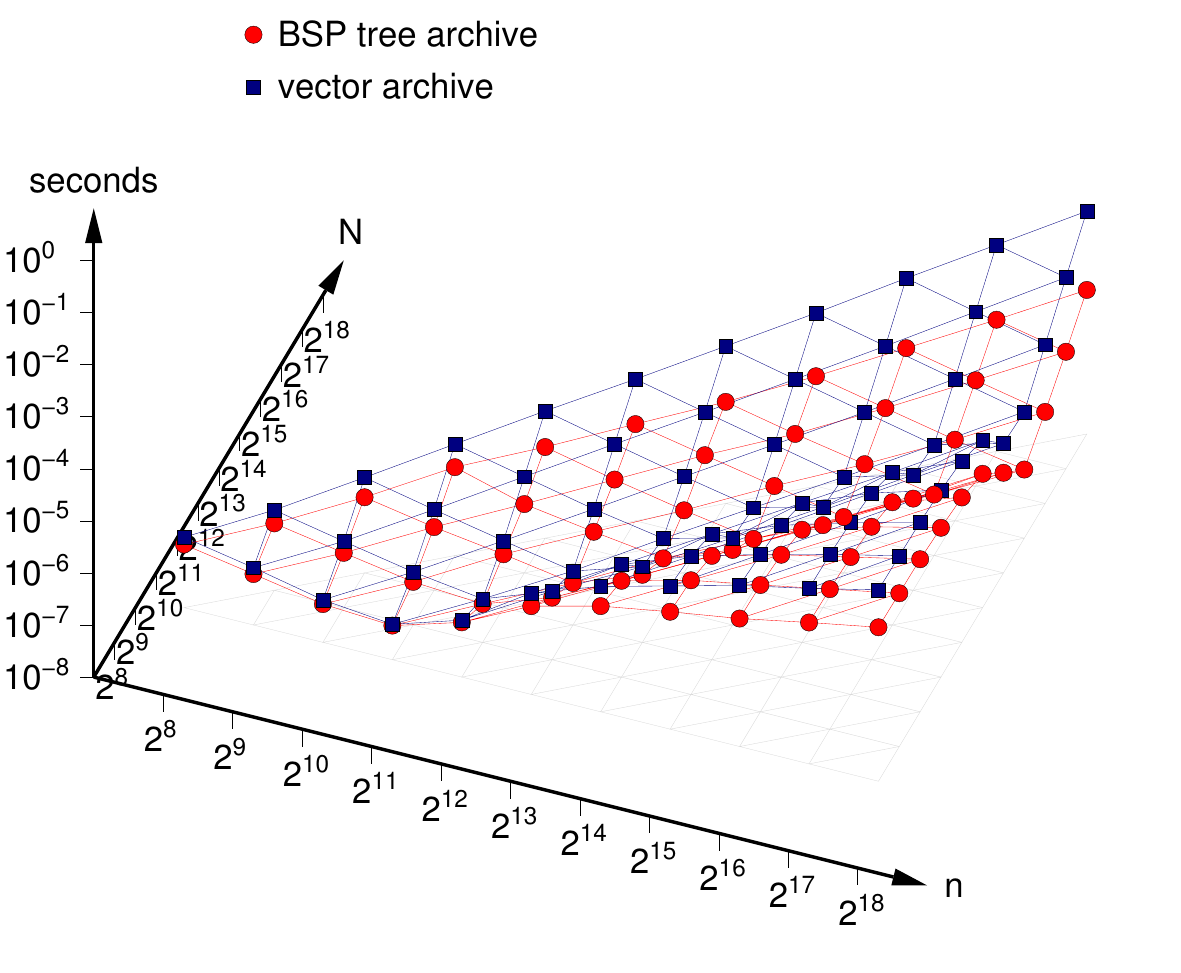}
}
\caption{ \label{figure:timings-2}
Processing time per objective vector for $m=5$ objectives (top) and
$m=10$ objectives (bottom), for a static distribution ($c=1$, left)
and solutions improving over time ($c=1.1$, right), for systematically
varied numbers of overall points ($n=N+D$) and non-dominated points ($N$)
in the range $2^8$ to $2^{18}$.
}
\end{center}
\end{figure}

We constructed archives from sequences of $D$ dominated ($d>0$) and $N$
non-dominated ($d=0$) normally distributed objective vectors according
to
\begin{align*}
	y^{(k)} \sim \Normal \left( \frac{d(N+D)}{k} \ones, \id - \frac{1}{m} \ones \ones^T \right)
	\enspace,
\end{align*}
where $\id \in \R^{m \times m}$ denotes the identity matrix and
$\ones = (1, \dots, 1)^T \in \R^m$ is the vector of all ones. The
distribution has unit variance in the subspace orthogonal to the $\ones$
vector.
The parameter $d \geq 0$ controls the systematic improvement of points
over time. At position $k$ of the sequence a dominated point was sampled
with probability $c\frac{D'}{N'+D'}$, where $D'$ and $N'$ denote the
number of remaining dominated and non-dominated points to be placed into
the sequence. Hence for $c > 1$ there is a preference for observing more
dominated points early on in the sequence, while for $c = 1$ there is
not.

The bucket size~$B$ and the tree re-balancing threshold~$z$ are tuning
parameters of the algorithm. We propose the settings $B=20$ and $z=6$ as
default values since they gave robust results across problems with
varying characteristics in preliminary experiments. The relatively large
bucket $B$ size yields well balanced trees. Therefore a high value of
the threshold~$z$ is affordable, because re-balancing is rarely needed.

Figures \ref{figure:timings-1} and \ref{figure:timings-2} display the
average processing times of the different archives over sequences with
varying $m$, $n=N+D$, $N$, and $c$. It is no surprise that for $m=2$ the
specialized bi-objective archive performs clearly best. For $m \geq 3$
the BSP tree is in all cases superior to the baseline. The vector
archive is only competitive for $N \ll n$, i.e., as long as the number
of non-dominated points in the archive remains small. Unsurprisingly,
this is also the domain where the systematic improvement of points over
time has a significant effect on archive performance. The overall effect
is similar for all archive types, and in comparison to the static case
the BSP archive can even increase its advantage over the linear memory
archive.

\begin{figure}
\begin{center}
{
	\includegraphics[scale=0.6]{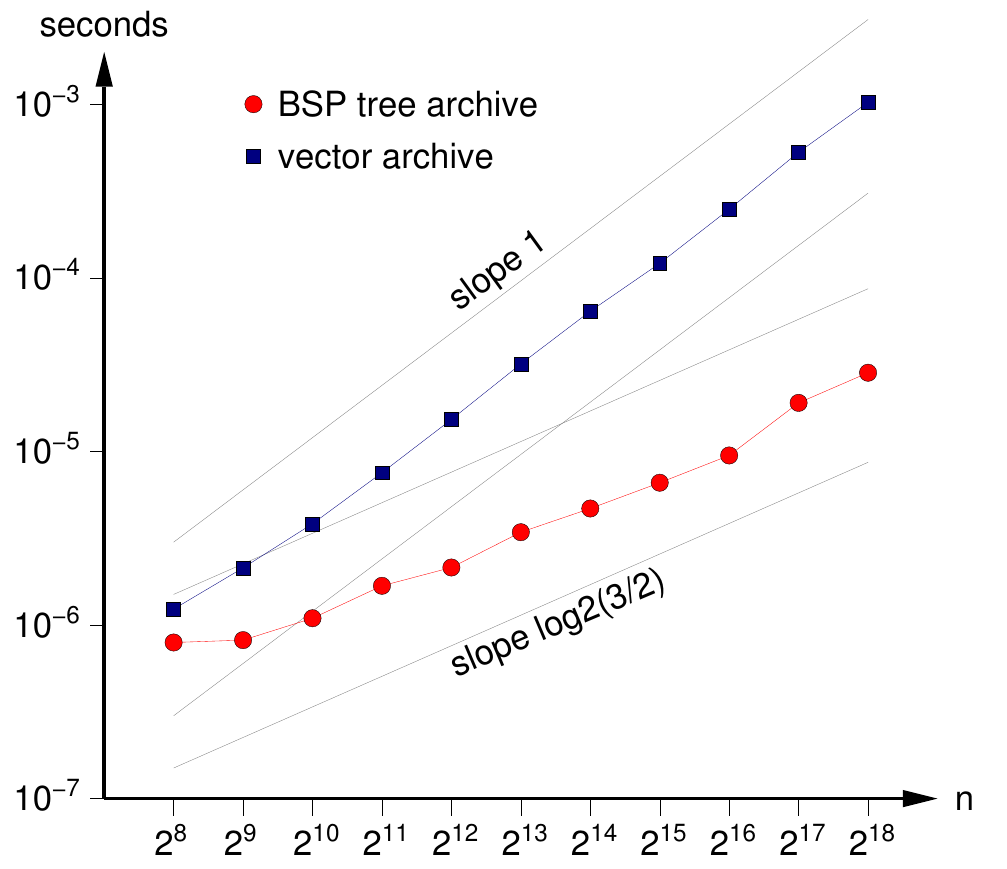}\includegraphics[scale=0.6]{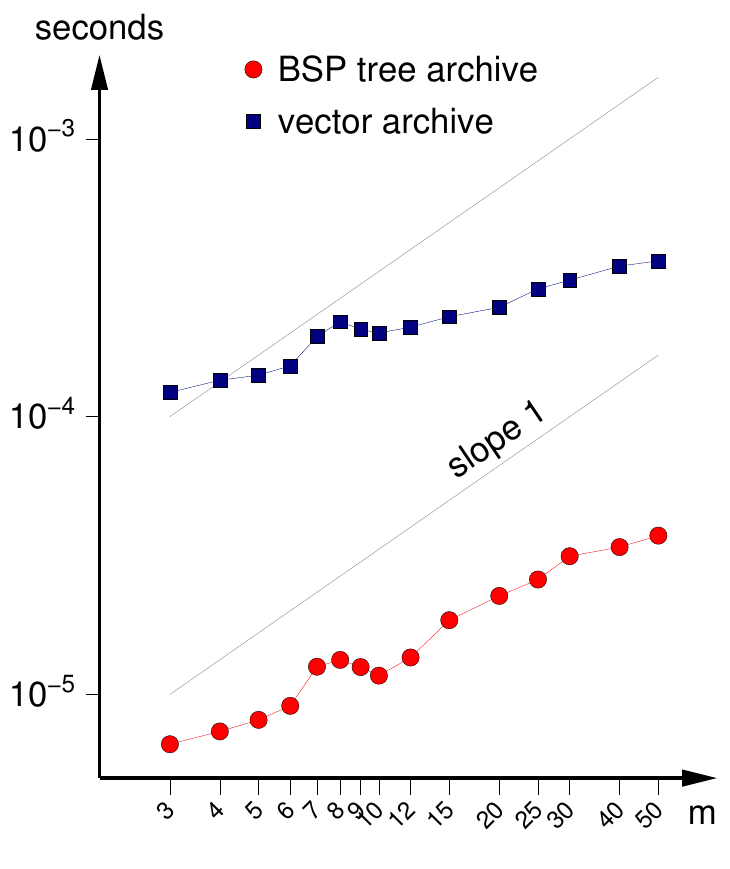}
}
\caption{ \label{figure:scaling}
Empirical scaling w.r.t.\ $n$ (left) and $m$ (right). The gray lines in
the background of the log-log-plots indicate the exponents $\alpha$ and
$\beta$ of the hypothetical scaling laws $T(n) \in \Theta(n^\alpha)$ and
$T(m) \in \Theta(m^\beta)$, respectively.
}
\end{center}
\end{figure}

Figure~\ref{figure:scaling} (left) shows the empirical scaling of the
archives for a setting close to the preconditions of the theoretical
analysis:
$m=3$, $D=0$, $c=1$. The actual scaling is very close to the lower bound
of order $n^{\log_2(3/2)} \approx n^{0.585}$ from
theorem~\ref{theorem:lower}. In contrast, the vector-based archive
scales perfectly linear in~$n$.

Figure~\ref{figure:scaling} (right) investigates the scaling to large
numbers of objectives. It plots the runtime per processing step for an
archive consisting of $2^{15}$ non-dominated points. In contrast to
Theorem~\ref{theorem:lower}, in practice the curve is of course not flat.
The algorithm still scales gracefully to large numbers of objectives.
For the range $3 \leq m \leq 50$ we observe sub-linear scaling. The
baseline method (surprisingly) exhibits even slightly better scaling
(while taking 10 times more time in absolute terms).

\subsection{MOEA Runs}

A second series of tests was performed with objective vectors generated
by state-of-the-art MOEAs on established benchmark functions. We used
two variants of the multi-objective covariance matrix adaptation
evolution strategy (MO-CMA-ES), namely with generational and with
steady-state selection \cite{igel:2007,voss:2010}. We applied the
implementations found in the Shark library,%
\footnote{\url{http://shark-ml.org}}
version 3.1 \cite{shark:2008}.
The population size was set to 100, all parameters were left at their
defaults. Sequences of objective vectors were generated by running the
optimizers on the scalable benchmark problems DTLZ1 to DTLZ4
\cite{deb:2002}, with 30 variables, 3 objectives, and a  budget of
200,000 function evaluations. The results are summarized in
table~\ref{table:EMOAs}.

\begin{table}
\begin{center}
\begin{tabular}{|l|l|r|r|r|r|}
\hline
~MOEA~ & ~problem~ & ~N~ & ~k-d~tree~ & ~vector~ & ~speed-up~ \\
\hline
~generational MO-CMA-ES~  & ~DTLZ1~ & ~104,191~ & ~10.70~ & ~167.78~ & ~15.7~ \\
~steady state MO-CMA-ES~  & ~DTLZ1~ &  ~91,022~ &  ~9.58~ & ~166.49~ & ~17.4~ \\
\hline
~generational MO-CMA-ES~  & ~DTLZ2~ &  ~42,153~ &  ~8.15~ &  ~64.71~ &  ~8.0~ \\
~steady state MO-CMA-ES~  & ~DTLZ2~ &  ~53,814~ & ~10.67~ &  ~79.58~ &  ~7.4~ \\
\hline
~generational MO-CMA-ES~  & ~DTLZ3~ &   ~7,895~ &  ~4.13~ &  ~32.93~ &  ~8.0~ \\
~steady state MO-CMA-ES~  & ~DTLZ3~ &  ~61,860~ &  ~4.59~ &  ~72.25~ & ~15.7~ \\
\hline
~generational MO-CMA-ES~  & ~DTLZ4~ &  ~11,621~ &  ~2.75~ &  ~13.31~ &  ~4.8~ \\
~steady state MO-CMA-ES~  & ~DTLZ4~ &  ~23,097~ &  ~3.86~ &  ~27.00~ &  ~7.0~ \\
\hline
\end{tabular}
\vspace*{0.5em}
\caption{ \label{table:EMOAs}
The table lists the number $N$ of non-dominated points at the end of
the optimization run, the average runtimes (in $10^{-6}$ seconds) for
processing a single point with the BSP tree archive and the vector
archive, as well as the speed-up factor of the BSP tree archive over the
vector archive.
}
\end{center}
\end{table}

In all cases the BSP tree archive was significantly faster than the
baseline. It outperformed the vector-based archive roughly by a factor
of 10. The exact speed-up correlates with the number of non-dominated
points. This is in line with our analysis, as well with the results for
analytically controlled distributions of objective vectors. The results
indicate that the tree-based archive also works well for realistic
sequences of objective vectors produced by MOEAs. It demonstrates its
strengths in particular if the set of non-dominated points grows large.

\section{Conclusion}
\label{section:conclusion}

We have presented an algorithm for updating an archive of Pareto optimal
objective vectors processed iteratively in a sequence. The data
structure is based on a k-d tree for binary space partitioning. This
choice yields runtime sub-linear in the archive size~$n$. The method is
shown to perform well for medium to large scale archives, where updating
performance matters most.
The archive is applicable to problems with arbitrary number $m$ of
objectives, including the many-objective case $m \gg 3$. Overall, these
properties make the proposed algorithm suitable for online processing of
non-dominated sets, e.g., in evolutionary multi-objective optimization.

\bibliographystyle{plain}

\end{document}